\documentclass[12pt,letterpaper,top=2.5cm,bottom=2.5cm,left=3cm,right=2.5cm,marginparwidth=1.75cm]{article}
\usepackage{arxiv}

\usepackage{graphicx}

\usepackage{caption}
\usepackage{subcaption}
\usepackage{textcomp}
\usepackage{amssymb}
\usepackage{amsthm}
\usepackage{amsmath}
\DeclareMathOperator{\sgn}{sgn}
\usepackage{url}
\usepackage{colortbl}
\usepackage{soul}
\usepackage{multirow}
\usepackage{float}
\usepackage{physics}
\usepackage{pifont}
\usepackage{color}
\usepackage{alltt}
\usepackage[hidelinks]{hyperref}
\usepackage{enumerate}
\usepackage{siunitx}
\usepackage{breakurl}
\usepackage{epstopdf}
\usepackage{pbox}
\usepackage{stfloats}
\usepackage[vlined,linesnumberedhidden,ruled,resetcount]{algorithm2e}
\usepackage{algpseudocode}  
\usepackage{cite}
\usepackage{amsmath, amsfonts}\usepackage{graphicx}

\usepackage{caption}
\usepackage{subcaption}
\usepackage{textcomp}
\usepackage{url}
\usepackage{colortbl}
\usepackage{multirow}
\usepackage{float}
\usepackage{physics}
\usepackage{pifont}
\usepackage{color}
\usepackage{alltt}
\usepackage[hidelinks]{hyperref}
\usepackage{enumerate}
\usepackage{siunitx}
\usepackage{breakurl}
\usepackage{epstopdf}
\usepackage{pbox}
\usepackage{stfloats}
\usepackage[vlined,linesnumberedhidden,ruled,resetcount]{algorithm2e}
\usepackage{algpseudocode}  
\usepackage{cite}
\usepackage{amsmath, amsfonts}

\allowdisplaybreaks

\newtheorem{propp}{Proposition}
\newtheorem{remm}{Remark}
\newtheorem{assm}{Assumption}

\usepackage{titlesec}
\usepackage{array}     
\usepackage{diagbox}   
\usepackage{booktabs}  

\begin{document}

\title{Robust Linear Quadratic Optimal Control of Cementitious Material Extrusion} 

\author{Mandana Mohammadi Looey, 
Amrita Basak, 
and Satadru Dey
\thanks{The authors are with the Department of Mechanical Engineering, The Pennsylvania State University, University Park, Pennsylvania 16802, USA. (e-mails: mfm6970@psu.edu, aub1526@psu.edu, skd5685@psu.edu).}
\thanks{This work was supported by National Science Foundation under Grants No. 2346650. The opinions, findings, and conclusions or recommendations expressed are those of the author(s) and do not necessarily reflect the views of the National Science Foundation.}}

\maketitle

\begin{abstract}                
Extrusion-based 3D printing of cementitious materials enables fabrication of complex structures, however it is highly sensitive to disturbances, material property variations, and process uncertainties that decrease flow stability and dimensional fidelity. To address these challenges, this study proposes a robust linear quadratic optimal control framework for regulating material extrusion in cementitious direct ink writing systems. The printer is modeled using two coupled subsystems: an actuation system representing nozzle flow dynamics and a printing system describing the printed strand flow on the build plate. A hybrid control architecture combining sliding mode control for disturbance rejection with linear quadratic optimal feedback for energy-efficient tracking is developed to ensure robustness and optimality. In simulation case studies, the control architecture guarantees acceptable convergence of nozzle and strand flow tracking errors under bounded disturbances. 
\end{abstract}



\section{Introduction}



Extrusion-based 3D printing, or direct ink writing (DIW), provides high-resolution deposition and geometric flexibility for cementitious materials, enabling the realization of intricate built-environment structures beyond the capabilities of traditional construction techniques. However, printability and resolution are strongly influenced by process parameters \cite{fasihi2024interaction} and ink rheology \cite{zhao2022effects},\cite{scalise2026multiphase}. Consequently, optimizing material properties and process conditions has been a major focus in recent research. For instance, Wang et al. evaluated printing quality based on filament uniformity, surface roughness, and cross-sectional shape to determine optimal printing conditions \cite{wang2025optimizing}.  
Hiremath et al. used a custom-built 3D printer to optimize material formulation and printing parameters for construction-scale 3D printing \cite{hiremath2025performance}. These studies improved shape retention, deposition consistency, and surface quality \cite{tu2023optimizing}, but they 
cannot adapt to variations in printing conditions in real time.



Incorporating real-time control into DIW systems addresses this limitation by enabling dynamic adjustment of process parameters, reducing manual intervention, and improving repeatability, accuracy, and efficiency in printing operations. Several closed-loop control strategies have been explored for extrusion-based 3D printing. Zomorodi and Landers implemented a hierarchical control structure combining ram velocity and extrusion force control in the Freeze-form Extrusion Fabrication process, showing improved robustness to variations in paste properties \cite{zomorodi2016extrusion}. Rabiei and Moini developed a feedback controller that regulates nozzle speed based on pressure feedback to maintain dimensional fidelity \cite{rabiei2025extrusion}. Ahi et al. combined printing parameters with rheological properties to control material flow rate in 3D concrete printing \cite{ahi2024automated} .


Despite the advances in closed-loop and optimization-based control for extrusion-based 3D printing \cite{zomorodi2016extrusion, rabiei2025extrusion, ahi2024automated, 9108592, guidetti2023data, guidetti2024data}, existing approaches often do not explicitly consider both real-time optimality and robustness in the printing and actuation systems. In our framework, optimality in the printing system is defined through two key criteria: first, the ability of the printer to accurately follow a prescribed reference flow velocity trajectory for the printed strand, minimizing deviations that could compromise layer uniformity and structural fidelity \cite{tu2023optimizing}; and second, minimizing the mechanical energy expended by the build plate, which is critical for efficient, cost-effective operation and reducing wear on printer components \cite{tu2023optimizing}.

Equally important is robustness, which ensures that both the printing system and the actuation system maintain desired performance in the presence of disturbances and uncertainties. In the printing system, disturbances may arise from variations in ink rheology, environmental conditions, or unexpected flow instabilities \cite{fasihi2024interaction, zhao2022effects}, and robustness ensures that the printed strand adheres to the target flow profile despite these factors. In the actuation system, uncertainties in nozzle response, motor dynamics, or sensor measurements can affect flow control \cite{fasihi2024interaction, zhao2022effects}, and robustness guarantees that the nozzle velocity closely tracks the reference trajectory without significant deviations. Through explicit integration of these optimality and robustness criteria, our approach provides a predictable, repeatable, and efficient printing process, enabling high-quality extrusion-based 3D printing of cementitious materials under variable operating conditions \cite{ahmed2023review, chen2022review}.  

In response to the aforementioned limitations, \textit{this study develops and analyzes a robust linear quadratic optimal control framework for extrusion-based 3D printing of cementitious materials -- which guarantees disturbance rejection in the actuation system and energy-efficient, high-accuracy trajectory tracking in the printing system.}




The structure of the paper is as follows: Section 2 presents the system model and formalizes the control problem. Section 3 details the robust linear quadratic control framework. Section 4 provides model validation, simulation results, and their analysis. Finally, Section 5 concludes the study and highlights future research directions.

\section{Robust Optimal Control of Extrusion-based Printer}

In this section, we discuss the system-level conceptualization of the 3D printer and its control.

\subsection{System Setup}

\begin{figure}[h!]
    \centering
    \includegraphics[width=0.5\textwidth]{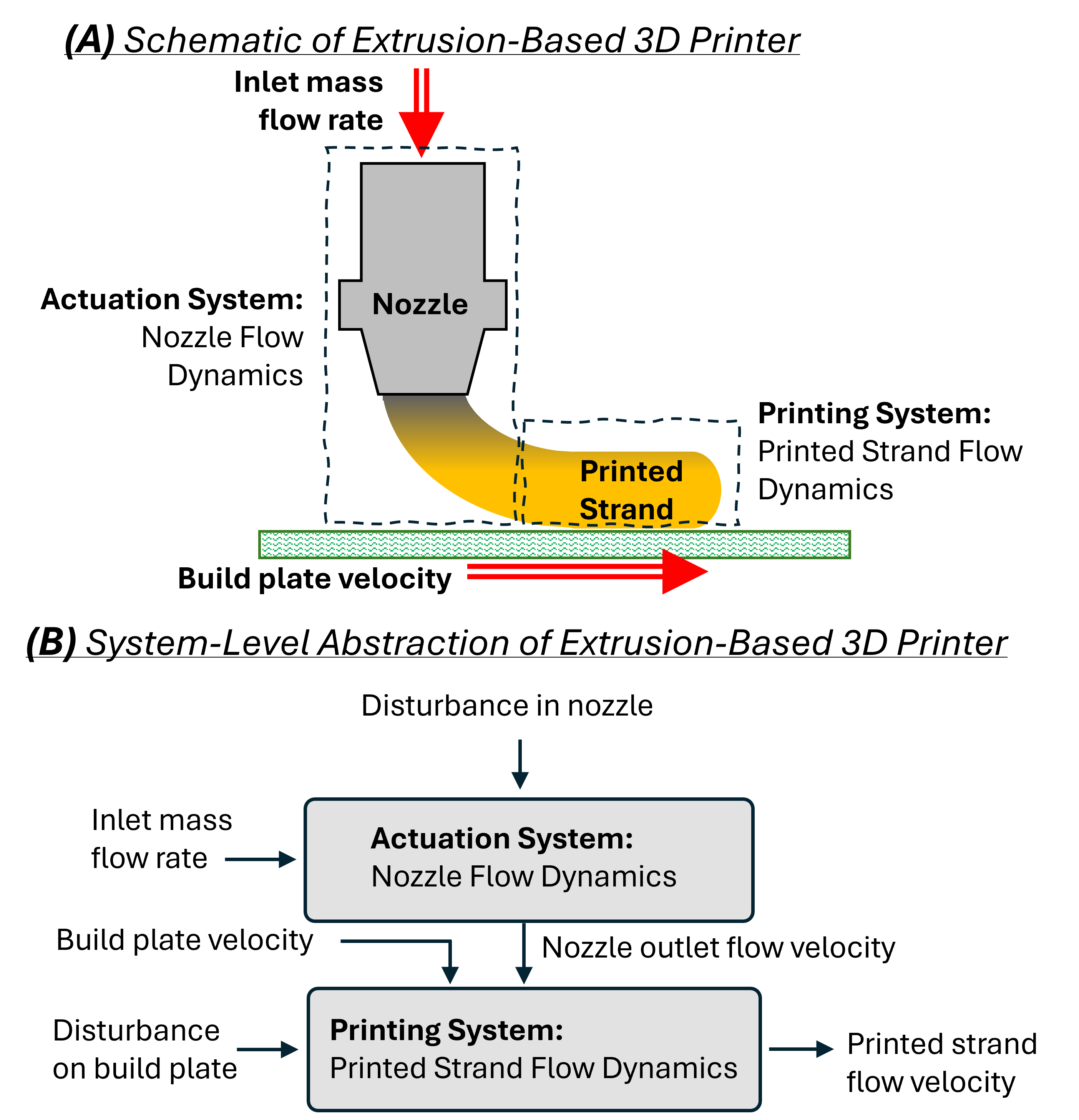}
    \caption{Schematic and system-level abstraction of extrusion-based 3D printer. }
    \label{fig:system_setup}
\end{figure}

A conceptual schematic of an extrusion-based 3D printer is shown in Fig. \ref{fig:system_setup}(A). The 3D printer can be conceptualized to have two underlying systems that work together: \textit{Actuation System} which represents the material flow dynamics in the nozzle depending on the inlet conditions, nozzle geometry, and the material rheology; and \textit{Printing System} which represents the flow dynamics in printed strand. A pressure-driven flow is created through the nozzle by a controlled inlet mass flowrate introduced into the nozzle. In the next stage, the material undergoes bending and swelling after exiting the nozzle and before deposition.

Fig. \ref{fig:system_setup}(B) shows the system-level abstraction of the 3D printer. Here, the \textit{Actuation System (AS)} receives the inlet mass flow rate as a control input signal while being affected by various disturbances in the nozzle. The \textit{Printing System (PS)} in turn receives the flow velocity from nozzle outlet and gets actuated by the moving plate velocity, while being affected by the disturbances in the plate area. The ultimate output of interest is the flow velocity in the printed strand.

\subsection{Mathematical Representation of Extrusion}
Based on our prior work on reduced order modeling of extrusion-based printers \cite{looey2025physics}, two cascaded systems shown in Fig. \ref{fig:system_setup}(B) can be represented by:
\begin{align}
    & \text{Actuation System: } \dot{x}_1 = A_1 x_1 + B_1 u_1 + \eta_1, \label{ss-1}\\
    & \text{Printing System: } \dot{x}_2 = A_2 x_2 + A_{21} x_1 + B_2 u_2 + \eta_2, \label{ss-2}
\end{align}
where $x_1$ and $x_2$ represent the states (material flow velocities) of the actuation system and the printing system, $u_1$ (inlet mass flow rate) and $u_2$ (build plate velocity) represent the inputs, $A_1$, $A_2$, $A_{21}$, $B_1$ and $B_2$ represent system matrices, and $\eta_1$ and $\eta_2$ represent disturbances.  


The dynamic state-space model \eqref{ss-1}-\eqref{ss-2} are derived from conservation of mass and momentum governing cementitious material flow through the nozzle and onto the build plate. In the actuation system, the inlet mass flow rate produces pressure-driven flow in the nozzle, where the outlet velocity is influenced by viscous resistance, nozzle geometry, and material rheology. The printing system then captures the evolution of the printed strand velocity, governed by the nozzle outlet flow and the build plate motion.

\subsection{Optimality and Robustness in Extrusion}

Here, we discuss the physical requirements and intuition for optimality and robustness in material extrusion.

\noindent \textbf{Optimality in the Printing System:} We consider two optimality measures in 3D printing performance: one arises from the printer's ability to follow a prescribed reference flow velocity trajectory for the printed strand with minimum error, and the other one originates from the need to minimize the applied input energy to the printer. {Minimization of the reference tracking error helps to avoid under/over extrusion and to maintain geometrical fidelity.}  

\noindent \textbf{Robustness in the Printing System:} As captured by the term $\eta_2$ in \eqref{ss-2}, the printing system is subject to various disturbances and uncertainties. {These disturbances and uncertainties represent variations in the printing process that cannot be perfectly predicted or controlled, such as partial nozzle clogging or unintended movement or vibrations of the build plate.} 

\noindent \textbf{Robustness in the Actuation System:} As captured by the term $\eta_1$ in \eqref{ss-1}, the actuation system is also subject to various disturbances and uncertainties. {Physically, these represent imperfections in the mechanisms that drive material extrusion, such as fluctuations in motor torque, actuator response delays, or variability in the extrusion drive. } Robustness in this case essentially means to follow the prescribed nozzle flow velocity behavior without being significantly affected by these disturbances and uncertainties. 

\section{Robust Linear Quadratic Optimal Control Framework}

In this section, we discuss the proposed robust optimal control approach.

\noindent \textbf{Robust Optimal Control Problem:} \textit{Given the mathematical representation of the 3D printer \eqref{ss-1}-\eqref{ss-2}, and prescribed reference nozzle flow velocity trajectory $x_{1_r}$ and reference printed strand flow velocity trajectory $x_{2_r}$, our objective is to find robust optimal control laws for the inlet mass flow rate $u_1 = f_1(x_1)$ and plate velocity $u_2 = f_2(x_1,x_2)$, such that the optimality and robustness criteria in Section 2.3 are satisfied.}

\subsection{Proposed Robust Optimal Control Architecture}
The proposed architecture of the robust optimal controller is shown in Fig. \ref{fig:control_dia}. The control system receives reference flow velocity profiles along with feedback signals (either measured by sensors or estimated by a real-time estimator), and in turn computes the control signals to achieve the control objectives. The control system consists of two control algorithms: 
\begin{itemize}
    \item \textit{Nozzle flow controller:} This controller operates on the error between reference nozzle flow velocity and the actual nozzle flow velocity, and computes the control signal of inlet mass flow rate. 
    \item \textit{Printed strand flow controller:} This controller operates on the error between reference printed strand flow velocity and the actual strand flow velocity, and computes the control signal of build plate velocity.
\end{itemize}

\begin{figure*}[h!]
    \centering
    \includegraphics[width=\textwidth]{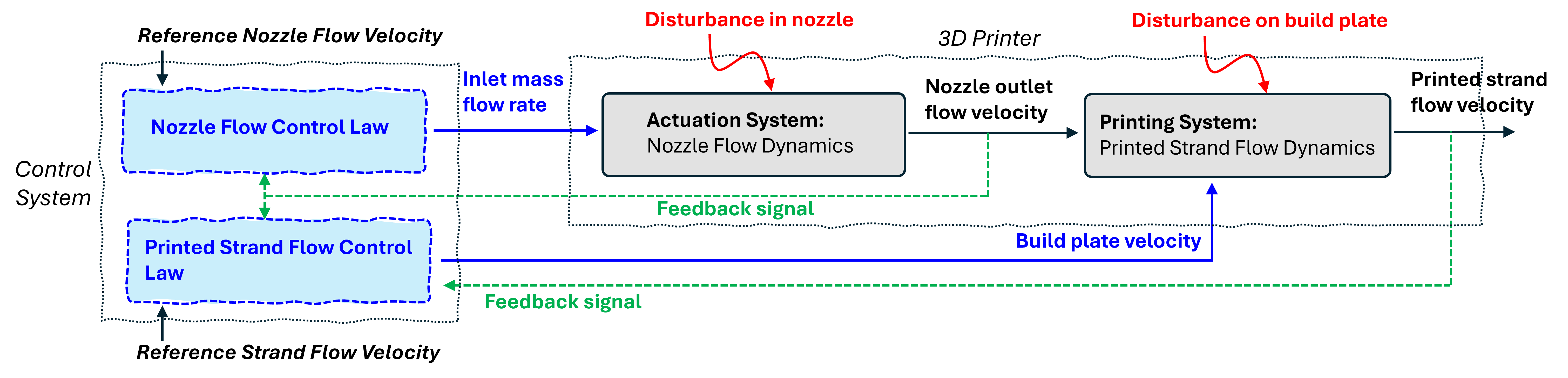}
    \caption{Proposed robust optimal control architecture. }
    \label{fig:control_dia}
\end{figure*}

\subsection{Design and Analysis of the Robust Optimal Controller}

In this subsection, we discuss the design approach of the controllers and provide a theoretical analysis of their performance. First, we make the following assumptions regarding the disturbances and reference trajectories:

\vspace{1mm}

\begin{assm}
    The uncertainty terms in \eqref{ss-1} and \eqref{ss-2} are bounded by the following:
    \begin{align}
        & \left|\eta_1\right| \leqslant \bar{\eta}_1, \ \left|\eta_2\right| \leqslant \bar{\eta}_2, \label{bnd}
    \end{align}
    where $\bar{\eta}_1,\bar{\eta}_2$ are constants known apriori. Furthermore, the reference trajectories and their derivatives are also bounded by the following:
    \begin{align}
        & \left|x_{1_r}\right| \leqslant \bar{x}_{1_r}, \left|\dot{x}_{1_r}\right| \leqslant \bar{x}_{1_{rd}}, \left|x_{2_r}\right| \leqslant \bar{x}_{2_r}, \left|\dot{x}_{2_r}\right| \leqslant \bar{x}_{2_{rd}}, \label{bnd-2}
    \end{align}
    where $\bar{x}_{1_r},\bar{x}_{1_{rd}},\bar{x}_{2_r},\bar{x}_{2_{rd}}$ are constants known apriori.
\end{assm}

\vspace{1mm}


Next, we choose the following mathematical forms of the control laws:
 \begin{align}
        &  u_1 = K_1 \sgn(x_{1_r}-x_1), \label{con-lw-1}\\
        &   u_2 = -K_{21}x_1 + u_{\mathrm{SM}} + u_{\mathrm{OPT}}, \label{con-lw-2} 
    \end{align}
    where \eqref{con-lw-1} represents \textit{nozzle flow control law} and \eqref{con-lw-2} represents \textit{printed strand flow control law}; $u_{\mathrm{SM}} =  K_{22}\,\mathrm{sgn}(s)$ with the sliding surface $s = \tilde{x}_2 - \int_{0}^{t}(A_2 - {B_2^{2}P}/{R})\tilde{x}_2\,d\tau$ and $u_{\mathrm{OPT}} = ({B_2P}/{R})\,\tilde{x}_2$; ${\tilde{x}}_2 = {x}_{2r} - {x}_2$; $K_1,K_{21},K_{22},R,P$ are control parameters to be determined; and $\sgn(a)=\left| a \right|/a, \forall a \neq 0$ is the signum function.

\vspace{1mm}

    \begin{remm}
        Note that the \textit{nozzle flow control law} \eqref{con-lw-1} consists of a sliding mode control term ($K_1 \sgn(x_{1_r}-x_1)$). This is because in the nozzle flow dynamics (actuation system), we are mainly concerned with robustness to disturbances (as discussed in Section 2.3). Hence, sliding mode becomes a suitable choice due to its robustness properties \cite{utkin2017sliding}. The mathematical structure of \textit{printed strand flow control law} \eqref{con-lw-2} on the other hand consists of a linear cancellation component ($-K_{21}x_1$), a linear optimal feedback component ($u_{\mathrm{OPT}}$), and a sliding mode component ($u_{\mathrm{SM}}$). This is due to the fact that we require both optimality and robustness properties for this printing system. Hence, inspired by \cite{young1997sliding} and \cite{5580607}, we choose such a combination law.
    \end{remm}

Now, we present the following mathematical formulation of the optimality and robustness objectives in Section 2.3.

\begin{itemize}
   \item \textit{Optimality in the Printing System:} This criterion is mathematically equivalent to minimizing the following cost functional: 
     \begin{align}
        \min_{u_2} \int_{0}^{t} \big[ Q (x_{2_r}(\tau)-x_2(\tau))^2 + R u_2^2(\tau) \big] d\tau. \label{opt-1}
   \end{align}
   where $Q$ and $R$ are user-defined weights.
    \item \textit{Robustness in the Printing System:} This criterion is mathematically equivalent to
     \begin{align}
        \left|x_2(t)-x_{2_r}(t)\right| \rightarrow 0, \ \eta_2 \neq 0. \label{rob-2}
   \end{align}
    \item \textit{Robustness in the Actuation System:} This criterion is mathematically equivalent to
     \begin{align}
        \left|x_1(t)-x_{1_r}(t)\right| \rightarrow 0, \ \eta_1 \neq 0. \label{rob-3}
   \end{align}
\end{itemize}

Next, we present our first proposition on the tracking and robustness performance of the \textit{nozzle flow controller}.

\vspace{1mm}

\begin{propp} [Robustness of nozzle flow dynamics]
    Consider the nozzle flow dynamics given by \eqref{ss-1}, the nozzle flow control law \eqref{con-lw-1}, the existence of a non-zero disturbance with the bound \eqref{bnd}, and the reference trajectory $x_{1_r}$ with the bounds \eqref{bnd-2}. Then, as desired by the robustness criterion \eqref{rob-3}, the reference nozzle flow velocity tracking error $\left|x_1-x_{1_r}\right| \rightarrow 0$ in finite time, if the controller gain satisfies the following condition:
    \begin{align}
        K_1 > \frac{1}{B_1 } (\bar{x}_{1_{rd}} + \left|A_1\right| + \bar{x}_{1_r} + \bar{\eta}_1). \label{cond-1}
    \end{align}
\end{propp}
\begin{proof}
We define the tracking error as $\tilde{x}_1 = x_{1_r} - x_1$, and compute the tracking error dynamics as:
\begin{align}
  \dot{\tilde{x}}_1 & = \dot{x}_{1_r} - \dot{x}_1 \nonumber\\
  & =  \dot{x}_{1_r} - A_1 x_{1_r} + A_1 \tilde{x}_1 - B_1 u_1 - \eta_1.  \label{er-1} 
\end{align}
Plugging in the control law \eqref{con-lw-1} in \eqref{er-1}, we get
\begin{align}
    & =  \dot{x}_{1_r} - A_1 x_{1_r} + A_1 \tilde{x}_1 - B_1 K_1 \sgn(\tilde{x}_1) - \eta_1. \label{er-2} 
\end{align}
Now, consider a candidate Lyapunov function $W_1 = \frac{1}{2} \tilde{x}^2_1$. Taking the derivative of $W_1$ along the tracking error trajectory dynamics, we get:
\begin{align}
   &  \dot{W}_1 = \tilde{x}_1\dot{\tilde{x}}_1 \nonumber \\
    & = \tilde{x}_1\dot{x}_{1_r} - A_1 \tilde{x}_1 x_{1_r} - \tilde{x}_1\eta_1 + A_1 \tilde{x}^2_1 - B_1 K_1 \tilde{x}_1\sgn(\tilde{x}_1). \label{lyap-1}
\end{align}
Now, applying Holder's inequality ($ab \leqslant \left|a b\right| \leqslant \left|a\right| \left|b\right|$) on the first three terms on the right hand side of \eqref{lyap-1}, considering the fact $\tilde{x}_1\sgn(\tilde{x}_1) = \left| \tilde{x}_1 \right|$, and $A_1 <0 $ due to the stable nature of the system, we can majorize the right hand side of \eqref{lyap-1} and get:
\begin{align}
   &  \dot{W}_1 \leqslant \left|\tilde{x}_1\right|\left|\dot{x}_{1_r}\right| + \left|A_1\right| \left| \tilde{x}_1\right| \left| x_{1_r}\right| + \left|\tilde{x}_1\right| \left|\eta_1\right| - B_1 K_1 \left|\tilde{x}_1\right| . \label{lyap-2}
\end{align}
Further applying the bounds \eqref{bnd}-\eqref{bnd-2}, we get
\begin{align}
     \dot{W}_1 & \leqslant \left|\tilde{x}_1\right| \bar{x}_{1_{rd}} + \left|A_1\right| \left| \tilde{x}_1\right| \bar{x}_{1_r} + \left|\tilde{x}_1\right| \bar{\eta}_1 - B_1 K_1 \left|\tilde{x}_1\right| \nonumber\\
   &  = \left|\tilde{x}_1\right| (\bar{x}_{1_{rd}} + \left|A_1\right| + \bar{x}_{1_r} + \bar{\eta}_1 - B_1 K_1). \label{lyap-3}
\end{align}
If the control gain $K_1$ is chosen to satisfy the condition \eqref{cond-1}, we have $\dot{W}_1 \leqslant - \alpha_1  \left|\tilde{x}_1\right|$ where $\alpha_1 >0$ is a positive representative lower bound of the quantity $B_1 K_1 - (\bar{x}_{1_{rd}} + \left|A_1\right| + \bar{x}_{1_r} + \bar{\eta}_1)$. This means that we have $\dot{W}_1 \leqslant - \sqrt{2}\alpha_1  \sqrt{W_1}$ which leads to the fact that $W_1 \rightarrow 0$ and hence $\left|\tilde{x}_1 \right| \rightarrow 0$ in finite time, even in the presence of disturbance $\eta_1 \neq 0$.
\end{proof}

Now, we present our proposition on the robust and optimal tracking performance of the \textit{printed strand flow controller}.

\vspace{1mm}

\begin{propp} [Optimality and robustness of strand flow dynamics]
    Consider the nozzle flow dynamics given by \eqref{ss-2}, the nozzle flow control law \eqref{con-lw-2}, the existence of a non-zero disturbance with the bound \eqref{bnd}, and the reference trajectory $x_{2_r}$ with the bounds \eqref{bnd-2}. Then, as desired by the robustness criterion \eqref{rob-2}, the reference printed strand flow velocity tracking error $\left|x_2-x_{2_r}\right| \rightarrow 0$ in finite time achieving the sliding motion defined by $s=0$ and $\dot{s}=0$, if the sliding mode controller gain $K_{22}$ satisfies the following condition:
    \begin{align}
        B_2K_{22} > \bar{x}_{2_{rd}} + |A_2|\,\bar{x}_{2_r} + \bar{n}_2. \label{cond-2}
    \end{align}
    Furthermore, under sliding motion, the control law $u_{\mathrm{OPT}} = {B_2P}{R}^{-1}\,\tilde{x}_2$ will minimize the objective function \eqref{opt-1}, where $P$ comes from the algebraic Riccati equation
\begin{equation}
2A_2P + Q = {P^{2}B_2^{2}}{R}^{-1}. \label{ric-1}
\end{equation}    
\end{propp}
\begin{proof}
We approach this proof following the technical treatment in \cite{young1997sliding} and \cite{5580607}. Considering \eqref{ss-2} and choosing $K_{21} = {A_{21}}/{B_2}$, the dynamics of the tracking error $\dot{x}_{2r} - \dot{x}_2$ can be written as:
\begin{align}
 \dot{\tilde{x}}_2 = \dot{x}_{2r} - \dot{x}_2  = & \ \dot{x}_{2r} - A_2 x_{2r} + A_2\tilde{x}_2  \nonumber\\
&  - B_2K_{22}\,\mathrm{sgn}(s) - B_2u_{\mathrm{OPT}} - n_2. \label{pr-2}
\end{align}
Consider the Lyapunov function candidate $W_2 = (1/2)s^2$ for the sliding surface $s$ as defined in \eqref{con-lw-2}. The derivative of this function can be written as:
\begin{align}
& \dot{W}_2 = s\dot{s} = s\left\{\dot{x}_{2r} - A_2x_{2r} - n_2 - B_2K_{22}\,\mathrm{sgn}(s)\right\}. \label{pr-22}
\end{align}
Applying Holder's inequality on some terms of the right hand side of \eqref{pr-22}, we get:
\begin{equation}
\dot{W}_2 \le -B_2K_{22}|s| + |s|\,\bar{x}_{2r,d} + |s|\,|A_2|\,\bar{x}_{2r} + \bar{n}_2|s|.
\end{equation}
Choosing $K_{22}$ such that \eqref{cond-2} is satisfied, gives us $\dot{W}_2 \leqslant - \lambda_1  \sqrt{W_2}$ (with some $\lambda_1>0$) which leads to the fact that $W_2 \rightarrow 0$ and hence $\left|{s} \right| \rightarrow 0$ in finite time, even in the presence of disturbance $\eta_2 \neq 0$ -- ultimately achieving the sliding motion \cite{utkin2017sliding}. In sliding motion, we have $s=0$ and $\dot{s}=0$ \cite{young1997sliding}. Now, considering $\dot{s}=0$, we can compute the equivalent control \cite{young1997sliding}:
\begin{equation}
\{K_{22}\,\mathrm{sgn}(s)\}_{\mathrm{eq}} \!=\! u_{\mathrm{eq}} \!=\! ({1}/{B_2})\left\{\dot{x}_{2r} - A_2x_{2r} - n_2\right\}. \label{eq-c}
\end{equation}
Plugging in \eqref{eq-c} in \eqref{pr-2}, we get:
\begin{align}
& \dot{\tilde{x}}_2
= \dot{x}_{2r} - A_2x_{2r} + A_2\tilde{x}_2 - B_2u_{\mathrm{OPT}} - n_2 - B_2u_{\mathrm{eq}} \nonumber\\
\Rightarrow & \dot{\tilde{x}}_2 = A_2\tilde{x}_2 - B_2u_{\mathrm{OPT}}. \label{eq-o}
\end{align}
The optimal control law for the system \eqref{eq-o} which minimizes the quadratic objective function
\begin{equation}
J = \int_{0}^{t}\left(Q\tilde{x}_2^{\,2} + Ru_{\mathrm{OPT}}^{\,2}\right)\,d\tau,
\end{equation}
is given by $u_{\mathrm{OPT}} = {R}^{-1}B_2P\,\tilde{x}_2$ where $P$ comes from the solution of algebraic Riccati equation \eqref{ric-1} \cite{lewis2012optimal}.
\end{proof}

\section{Results and Discussion}

This section evaluates the performance of the optimal and robust controllers. A comprehensive plant model that incorporates input-dependent parameters and considers realistic uncertainties is used in this study. The model uncertainties are derived from the probability distribution generated from the data collected in \cite{looey2025physics}. We evaluate the robust optimal controller performance through the following case studies.

 \begin{figure*}[h!]
    \centering
    \includegraphics[width=0.5\textwidth]{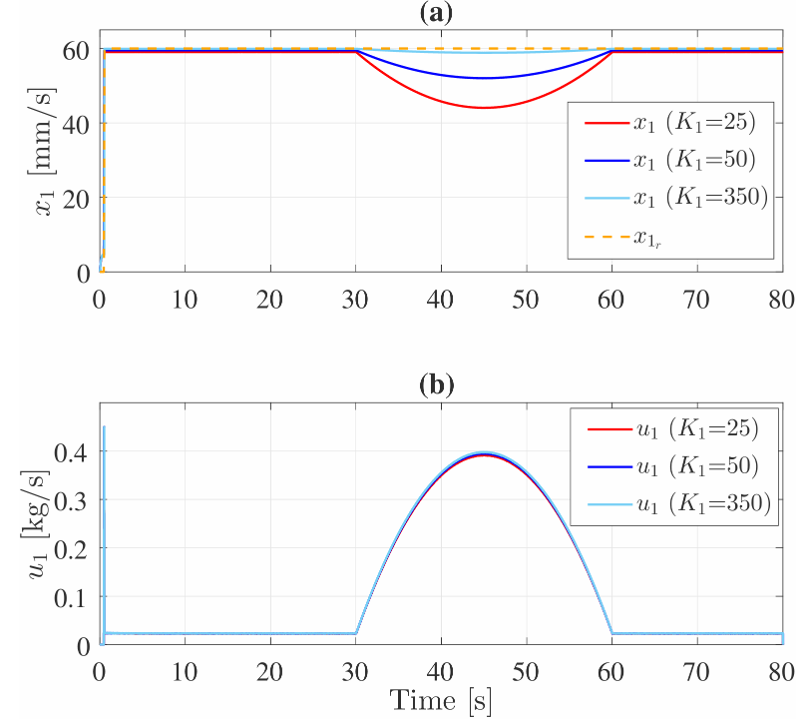}
    \caption{Reference velocity tracking for the actuation system under disturbance injection.}
    \label{fig:actuation system}
\end{figure*}

\textbf{Case Study 1:} The first case (shown in Fig. \ref{fig:actuation system}) analyzes the velocity tracking performance in the actuation system. A reference velocity is provided to the controller, which was attained in the first few seconds subsequently maintaining a smooth tracking. Next, a quadratic disturbance profile is introduced to the system between $t=30s$ and $t=60s$ to simulate temporary degradation of the actuation flow. As shown in Fig.3(a), the disturbance causes the nozzle flow velocity to deviate from the reference trajectory.
In response to this, the controller increases the inlet mass flow rate to compensate for the error (Fig.\ref{fig:actuation system}(b)). It is observed that the tracking recovery rate depends on the controller gain $K_1$, introduced in \ref{con-lw-1}. Increasing $K_1$ increases the corrective control action, which suppresses the disturbance effect and reduces tracking error. 
It is confirmed by Fig. \ref{fig:error_AS_PS}(a) which shows that with increasing control gain to $K_1 = 350$, tracking error is minimized to less than $2\%$. Furthermore, a higher peak of the control input in this plot indicates a more aggressive disturbance rejection mechanism of the controller.


\begin{figure*}[h!]
    \centering
    \includegraphics[width=0.5\textwidth]{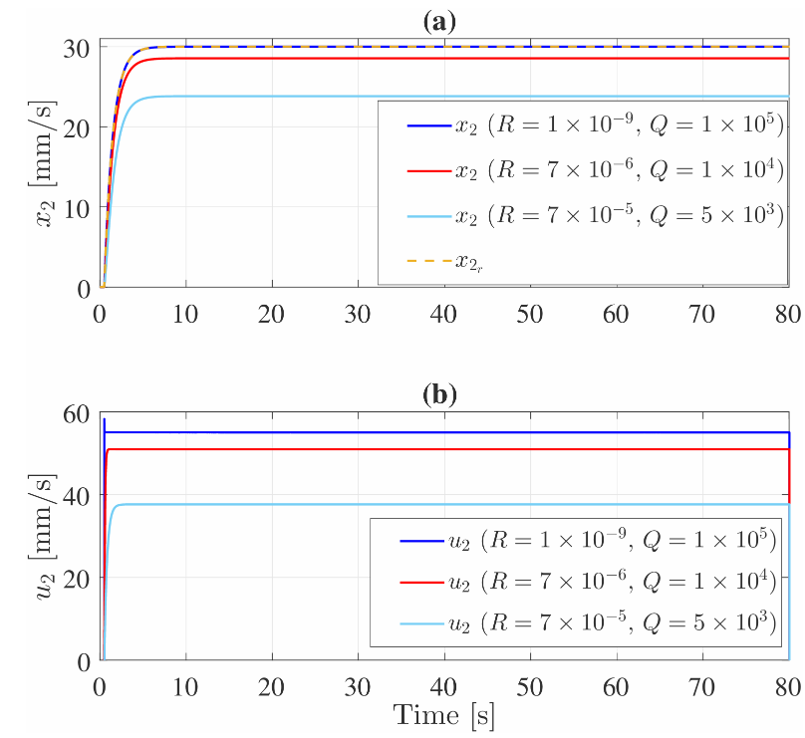}
    \caption{Reference velocity tracking for printing system with optimal only controller.}
    \label{fig:printing system}
\end{figure*}

\textbf{Case Study 2:} The second case (shown in Fig. \ref{fig:printing system}) evaluates the performance of the optimal only controller in the printing system (without the sliding mode component). 
We considered different combinations of control weights $R$ and $Q$ to find the optimal case for an acceptable tracking error. 
Fig. \ref{fig:printing system} shows how changing $Q$ and $R$ affects the tracking performance. In Fig. \ref{fig:printing system}(a), although all cases reach stable steady-state, only certain pairs of control weights with smaller $R$ and larger $Q$ exhibit faster convergence to the reference with reduced steady-state error. Furthermore, Fig.\ref{fig:printing system}(b) indicates a more aggressive controller input command with increasing $Q$ and decreasing $R$. 

\begin{figure}[h!]
    \centering
    \includegraphics[width=0.5\textwidth]{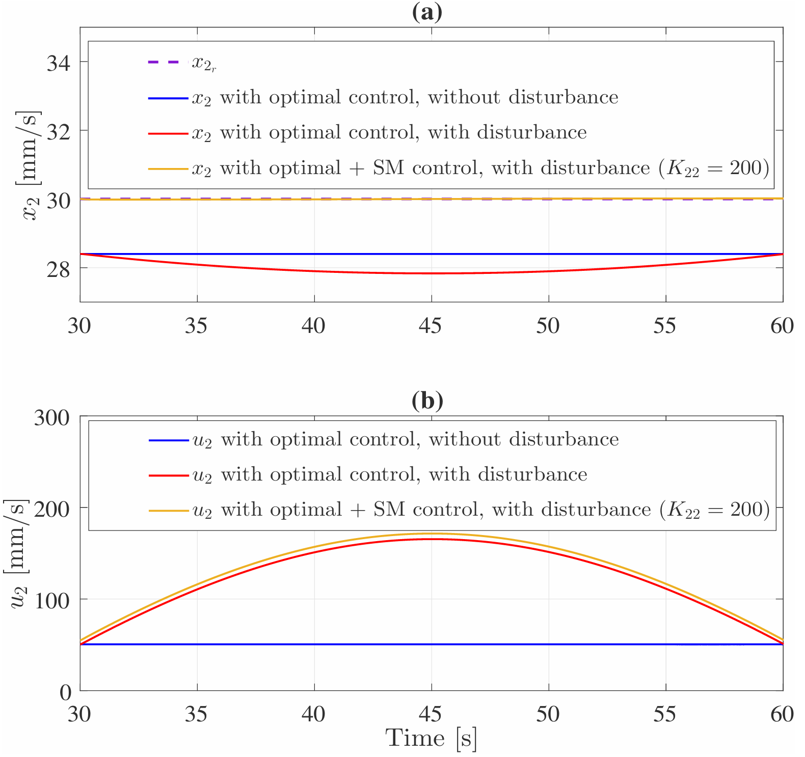}
    \caption{Reference velocity tracking for printing system with optimal and sliding mode (SM) controller.}
    \label{fig:printing system_sm_opt}
\end{figure}

\textbf{Case Study 3:}  The third case (shown in Fig. \ref{fig:printing system_sm_opt}) investigates the effect of incorporating optimal control and sliding mode control together into the algorithm for the printing system. Fig. \ref{fig:printing system_sm_opt}(a) shows that in the absence of disturbances, the optimal controller maintains steady-state tracking of the reference velocity with a constant steady-state error. Next, a quadratic disturbance profile is introduced between $t=30s$ and $t=60s$ to simulate the effect of build plate motion issues. When the disturbance is injected into the system, the state deviates from the reference, increasing the tracking error -- before gradually returning toward the steady-state value once the disturbance subsides. As shown in Fig. \ref{fig:printing system_sm_opt}(b), the control effort increases in response to the disturbance; however, the optimal controller alone does not significantly reduce the tracking error. Therefore, sliding mode control is incorporated to enhance disturbance rejection. 
Figs. \ref{fig:printing system_sm_opt}(a) and (b) demonstrate that the combined optimal–sliding mode controller significantly improves tracking performance while requiring only a modest increase in control effort 
(approximately $3\%$).


Figs. \ref{fig:error_AS_PS}(a) and (b) illustrate the influence of sliding mode controller gain ($K_1$) in actuation system and the optimal control weighting matrices $Q$ and $R$ in printing system. The effects on these two systems are evaluated in terms of the percentage maximum and percentage steady-state tracking error, respectively. Fig. \ref{fig:error_AS_PS}(a) shows that higher sliding mode controller gain leads to lower tracking error during disturbance injection. Fig. \ref{fig:error_AS_PS}(b) highlights the inherent trade-off in the optimal controller design: increasing $Q$ improves tracking accuracy by penalizing state error more heavily, while decreasing $R$ increases control effort to enforce faster correction. 
The results demonstrate that for sufficiently large values of $Q$, the steady-state error decreases below $5\%$, indicating high-fidelity tracking performance. 
The surface trend confirms the inherent trade-off between tracking accuracy and control energy. Therefore, the optimal selection of $Q$ and $R$ must balance steady-state accuracy with practical constraints on actuator capability and energy consumption.

\begin{figure}[h!]
    \centering
    \includegraphics[width=0.6\textwidth]{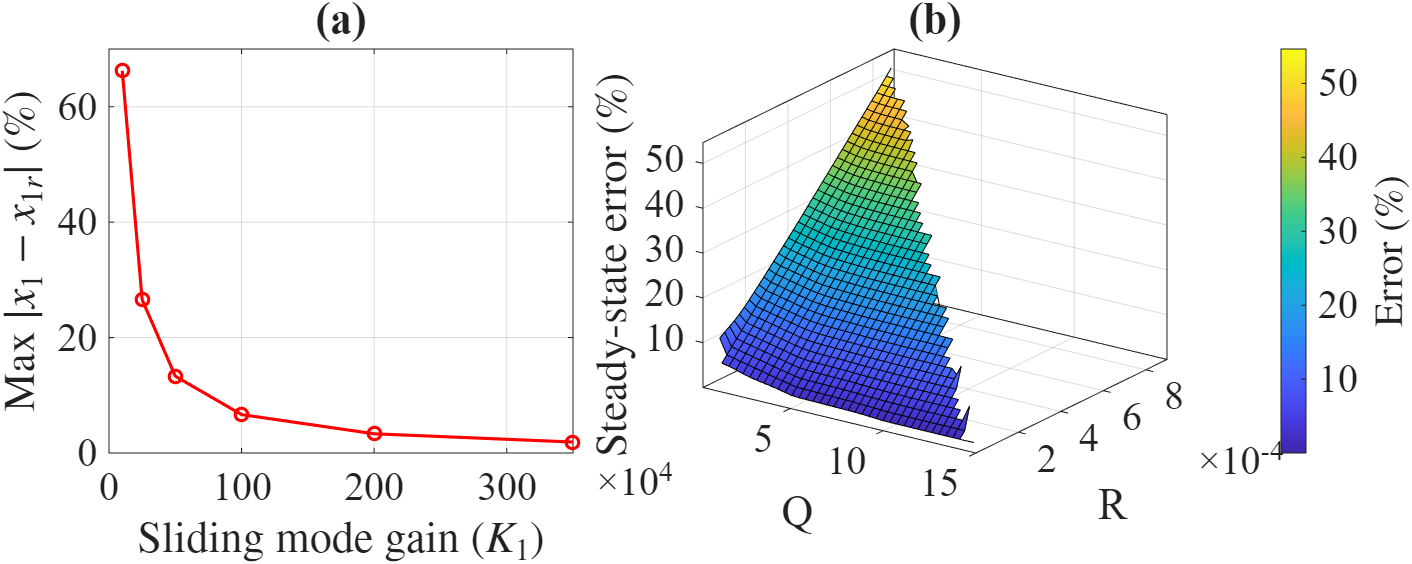}
    \caption{(a) Maximum tracking error after disturbance injection as a function of sliding mode gain in actuation system. (b) Percentage steady-state tracking error as a function of optimal control weights in printing system.}
    \label{fig:error_AS_PS}
\end{figure}


\section{Conclusions}

This study presented a robust linear quadratic optimal control framework to enhance the printability of cement-based materials in extrusion-based 3D printing. The proposed architecture integrates sliding mode control for disturbance rejection with linear quadratic optimal control for energy-efficient trajectory tracking. Robustness to disturbances in the actuation system is ensured by the sliding mode controller, while the printing subsystem achieves optimality and robustness in tracking performance through the linear quadratic control and sliding mode. Simulation results demonstrate improved disturbance rejection through increased sliding mode gain and highlight the trade-off between tracking accuracy and control effort when tuning optimal control weights. Overall, the framework improves flow regulation and extrusion fidelity under uncertainties. Future work will focus on experimental validation on a physical printer platform and extending the framework to coupled material processing and flow dynamics models.
\bibliographystyle{ieeetr}
\bibliography{ifacconf}             
                                                   







\end{document}